\documentclass[12pt]{article}

\usepackage[T1]{fontenc}
\usepackage[active]{srcltx}
\usepackage{lmodern}
\usepackage{amsmath,amstext,amssymb,graphicx,graphics,color}
\usepackage{mathtools,mathrsfs}
\usepackage{tikz}
\usetikzlibrary{decorations.markings,calc,fit,intersections,arrows,matrix,trees,positioning,shapes,tikzmark,backgrounds}
\usepackage{mleftright}
\usepackage{pgfplots}
\usepackage{theorem}
\usepackage{cite}               
\usepackage{hyperref}           
\usepackage{bbm}                
\usepackage{fancybox}           
\usepackage[section]{placeins}  
\usepackage[font=footnotesize]{caption}
\usepackage{subcaption}

\DeclarePairedDelimiter\floor{\lfloor}{\rfloor}
\DeclareMathOperator*{\argmin}{\arg\!\min}

\usepackage{arydshln} 
\tikzset{
    invisible/.style={opacity=0},
    visible on/.style={alt={#1{}{invisible}}},
    alt/.code args={<#1>#2#3}{%
      \alt<#1>{\pgfkeysalso{#2}}{\pgfkeysalso{#3}}%
  }
}
\theorembodyfont{\normalfont\slshape} 

\newtheorem{theorem}{Theorem}
\newtheorem{proposition}{Proposition}[section]

\newtheorem{example}[proposition]{Example}
\newtheorem{lemma}[proposition]{Lemma}
\theoremstyle{break} 

\newenvironment{proof}%
{{\par\noindent \bf Proof. \nobreak}}%
{\nobreak \removelastskip \nobreak \hfill $\Box$ \medbreak}
{{\par\noindent \bf Proof \nobreak}}%
{\nobreak \removelastskip \nobreak \hfill $\Box$ \medbreak}
{{\par\noindent \bf Proof lemma. \nobreak}}%
{\nobreak \removelastskip \nobreak \bf End proof lemma. \medbreak}
\newenvironment{remark}{\par \medskip \noindent {\bf Remark. }\nobreak}{\par \medskip}
\usepackage{fancyhdr}
 \fancypagestyle{plain}{
  \fancyhead{}
  
}

\def\paragraph#1{{\bf #1\ }}

\newcommand{\Var}{\mathrm{Var}}

\newcommand{\dd}{\mathrm{d}}

\usepackage[utf8]{inputenc}
\usepackage{unicode_character}

\title{From Gini index as a Lyapunov functional to convergence in Wasserstein distance}

\author{Fei Cao \footnotemark[1]}

\begin{document}
\maketitle

\footnotetext[1]{University of Massachusetts Amherst - Department of Mathematics and Statistics, Amherst, MA 01003, USA}

\tableofcontents

\begin{abstract}
In several recent works on infinite-dimensional systems of ODEs \cite{cao_derivation_2021,cao_explicit_2021,cao_iterative_2024,cao_sticky_2024}, which arise from the mean-field limit of agent-based models in economics and social sciences and model the evolution of probability distributions (on the set of non-negative integers), it is often shown that the Gini index serves as a natural Lyapunov functional along the solution to a given system. Furthermore, the Gini index converges to that of the equilibrium distribution. However, it is not immediately clear whether this convergence at the level of the Gini index implies convergence in the sense of probability distributions or even stronger notions of convergence. In this paper, we prove several results in this direction, highlighting the interplay between the Gini index and other popular metrics, such as the Wasserstein distance and the usual $\ell^p$ distance, which are used to quantify the closeness of probability distributions.
\end{abstract}


\section{Introduction and motivation}\label{sec:sec1}
\setcounter{equation}{0}

In recent years, infinite dimensional system of ODEs arising from the mean-field limits of stochastic agent-based models motivated from economic and social sciences are ubiquitous \cite{bassetti_mean_2015,cao_binomial_2022,cao_derivation_2021,cao_interacting_2022,cao_iterative_2024,cao_sticky_2024,cao_uncovering_2022,cao_uniform_2024,dragulescu_statistical_2000}.
Typically, infinite dimensional ODE systems investigated in the aforementioned work take the following generic form
\begin{equation}\label{eq:mean-field}
{\bf p}' = Q[{\bf p}],
\end{equation}
where ${\bf p} = (p_0,p_1,\ldots,p_n,\ldots)$ with ${\bf p}(t=0) \in \mathcal{P}(\mathbb N)$, and $Q \colon \mathcal{P}(\mathbb N) \to \mathcal{P}(\mathbb N)$ is a model-dependent operator (either linear or nonlinear). Moreover, it is often the case that the evolution system \eqref{eq:mean-field} preserves the total probability mass and the mean value, meaning that
\begin{equation}\label{eq:conservation_mass_mean_value}
\sum_{n=0}^\infty p'_n =0 \quad \text{and} \quad \sum_{n=0}^\infty n\,p'_n =0.
\end{equation}
Consequently, the solution ${\bf p}(t)$ lives in the space of probability distributions on $\mathbb N$ with the given mean value $\mu \coloneqq \sum_{n\geq 0} n\,p_n(0) \in \mathbb{R}_+$, defined by
\begin{equation}\label{eq:S_mu}
\mathcal{V}_\mu \coloneqq \left\{{\bf p} \mid \sum_{n=0}^\infty p_n =1,~p_n \geq 0,~\sum_{n=0}^\infty n\,p_n =\mu\right\}.
\end{equation}
We now briefly recall a series of relevant work below as motivating examples.

\begin{example}[The rich-biased dollar exchange model \cite{cao_derivation_2021}]\label{ex:1}
In many econophysics models where agents are only allowed to have non-negative integer-valued wealth (thus no debt is permitted), the mean-field ODE system is of the form \eqref{eq:mean-field}. In this context, $p_n(t)$ represents the probability that a typical agent has $n$ dollars at time $t$ and $\mu$ denotes the (conserved) average amount of dollars per agent. As a concrete example, the (mean-field) rich-biased exchange model introduced and investigated in \cite{cao_derivation_2021} is governed by the following nonlinear ODE system of infinite dimension:
\begin{equation}
\label{eq:Q_rich_biased}
p'_n = \left\{
      \begin{array}{ll}
        p_1-\overline{w}\,p_0 & \quad \text{for } n=0, \\
        \frac{p_{n+1}}{n+1} + \overline{w}\,p_{n-1} - \left(\frac{1}{n}+\overline{w}\right)\,p_n & \quad \text{for } n \geq 1,
      \end{array}
    \right.
  \end{equation}
in which $\overline{w} = \overline{w}[{\bf p}] \coloneqq \sum_{n =1}^\infty \frac{p_n}{n}$.
\end{example}

\begin{example}[The iterative persuasion-polarization opinion model \cite{cao_iterative_2024}]\label{ex:2}
In the context of opinion dynamics and in the mean-field region where the number of agents is sent to infinity, $p_n(t)$ represents the probability that a typical agent's opinion level (towards a given topic) at time $t$ equals to $n$, and $\mu$ denotes the (conserved) mean opinion among a large society. As a specific instance, the (mean-field) iterative persuasion-polarization opinion dynamics proposed and studied in \cite{cao_iterative_2024} is given by the following finite dimensional system of nonlinear ODEs (under certain parameter choices):
\begin{equation}\label{eqn:Q_ipp}
p'_n =  \left\{
      \begin{array}{ll}
p_0\,p_1 - p_0\,(1-p_0) & \quad \text{for } n=0,\\
p_{n-1}\,\sum_{j=n}^{2k} p_j + p_{n+1}\,\sum_{j=0}^{n} p_j  - p_n\,(1-p_n) & \quad \text{for } 0<n<2k, \\
p_{2k}\,p_{2k-1} - p_{2k}\,(1-p_{2k}) & \quad \text{for } n=2k,
\end{array}
    \right.
\end{equation}
where $k \in \mathbb{N}_+$ is arbitrary but fixed, ${\bf p}(t) = (p_0(t),\ldots,p_{2k}(t) \in \mathcal{P}(\{0,1,\ldots,2k\})$ and $\sum_{n=0}^{2k} n\,q_n(t) = \mu$ for all $t\geq 0$ with $\mu \in (0,2k)$.
\end{example}

\begin{example}[The sticky dispersion model on complete graphs \cite{cao_sticky_2024}]\label{ex:3}
In the context of dispersion processes (on complete graphs) introduced first in the literature on interacting particle systems \cite{cooper_dispersion_2018,de_dispersion_2023}, the investigation of the associated mean-field limit is very recent \cite{cao_quantitative_2024,cao_sticky_2024}. In the mean-field region where the total number of particles approaches to infinity, $p_n(t)$ represents the probability that a typical site hosts $n$ particles at time $t$ and $\mu$ denotes the (conserved) average number of particles per site. Now we recall the (mean-field) sticky dispersion model on complete graphs examined in \cite{cao_sticky_2024}:
\begin{equation}\label{eq:Q_sticky_dispersion}
p'_n = \left\{
    \begin{array}{ll}
      -\left(\mu-1+p_0\right)\,p_0 & \text{for } n=0, \\
      n\,p_{n+1}+\left(\mu-1+p_0\right)\,p_{n-1}- (n-1)\,p_n - \left(\mu-1+p_0\right)\,p_n & \text{for } n\geq 1.
    \end{array}
  \right.
\end{equation}
Here the most relevant region of the parameter $\mu >0$ for our follow-up discussion will be $\mu \in (0,1]$, as the large time behavior of the system \eqref{eq:Q_sticky_dispersion} changes drastically when $\mu > 1$ \cite{cao_sticky_2024}.
\end{example}

In several concrete mean-field ODE systems of the form \eqref{eq:mean-field} considered in \cite{cao_iterative_2024,cao_sticky_2024}, and in particular the situation encountered in Example \ref{ex:2} and Example \ref{ex:3} above, the so-called Gini index, defined via
\begin{equation}\label{def1:Gini}
G[{\bf p}] = \frac{1}{2\,\mu} \sum\limits_{i\in \mathbb N}\sum\limits_{j \in \mathbb N} |i-j|\,p_i\,p_j
\end{equation}
for any ${\bf p} \in \mathcal{P}(\mathbb N)$, is a Lyapunov functional along the solution of the system \eqref{eq:mean-field} for all $t\geq 0$. In other words, we have
\begin{equation}\label{eq:Gini_decay}
\frac{\dd}{\dd t} G[{\bf p}] \leq 0.
\end{equation}

\begin{remark}
The Gini index $G$ is a well-known inequality index which measures the (wealth) inequality of a given (one-dimensional) probability distribution and ranges from $0$ (for a wealth-egalitarian society) to $1$ (extreme inequality). Besides its economic origin, the notion of Gini index has also seen its application in many other fields other than economics, such as opinion/consensus models \cite{cao_iterative_2024,meng_fair_2023}, sparse representation of signals \cite{hurley_comparing_2009}, and the so-called classification and regression trees \cite{daniya_classification_2020}. Recently, extension of the classical Gini index to higher dimensions has also been explored \cite{auricchio_extending_2024} with applications to market economy.
\end{remark}

The unique equilibrium distribution, denoted by ${\bf p}^*$, associated to the mean-field system of ODEs provided in Example \ref{ex:2} and Example \ref{ex:3}, is given by the following shifted Bernoulli distribution supported on (at most) two spots:
\begin{equation}\label{eq:Bernoulli_equil}
p^*_{\floor*{\mu}} = 1-\mu + \floor*{\mu},~~~p^*_{\floor*{\mu}+1} = \mu - \floor*{\mu},~~~\text{and}~~~ p^*_n = 0 ~~\text{for $n \notin \{\floor*{\mu},1+\floor*{\mu}\}$},
\end{equation}
in which $\floor*{\mu}$ represents the integer part of $\mu$. The shifted Bernoulli equilibrium distribution ${\bf p}^*$ \eqref{eq:Bernoulli_equil} is actually a unique (global) minimizer of the Gini index \eqref{def1:Gini} over the set $\mathcal{V}_\mu$ of probability mass functions with mean $\mu$ \cite{cao_iterative_2024} (see Figure \ref{fig:1} for a basic illustration). This variational characterization of the Bernoulli-type (equilibrium) distribution ${\bf p}^*$ delivers a very clear economic intuition as explained in \cite{cao_iterative_2024}.
\begin{figure}[!htb]
\centering
\includegraphics[scale=0.6]{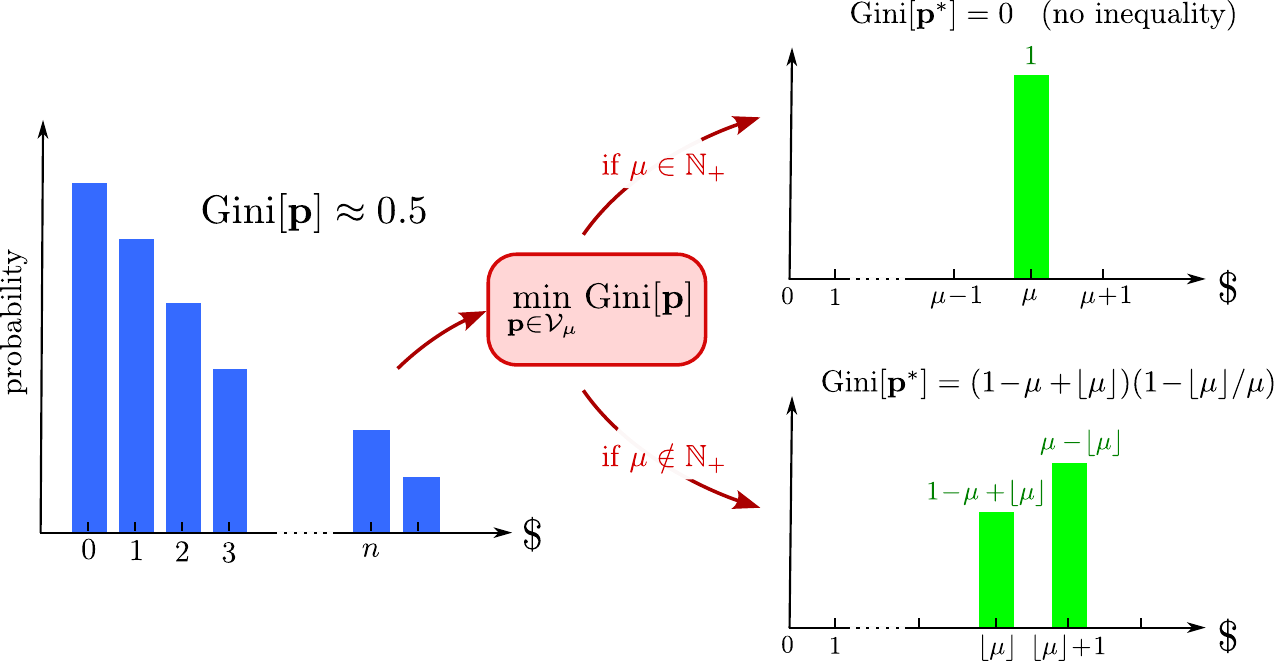}
\caption{Illustration of the variational characterization ${\bf p}^* = \argmin_{{\bf p} \in \mathcal{V}_\mu} G[{\bf p}]$ of the shifted Bernoulli-type distribution ${\bf p}^*$ \cite{cao_iterative_2024}. From an economic point of view, the distribution ${\bf p}^*$ represents the most egalitarian way (measured in terms of the Gini index) to distribute a large sum of money across a large population, subject to the constraints that each agent's wealth is a non-negative integer and the average wealth per agent equals some predetermined value $\mu \in \mathbb{R}_+$ .}
\label{fig:1}
\end{figure}
The key question which we aim to answer in this paper is: Assume that one can prove the convergence of the Gini index in the sense that
\begin{equation}\label{eq:conv_Gini}
G[{\bf p}(t)] - G[{\bf p}^*] = G[{\bf p}(t)] - \min\limits_{{\bf p} \in \mathcal{V}_\mu} G[{\bf p}] \xrightarrow{t \to \infty} 0.
\end{equation}
In which sense can we establish the convergence of ${\bf p}(t)$ towards its unique equilibrium ${\bf p}^*$~? Our main result in this manuscript can be summarized as follows:
\begin{theorem}\label{thm:main}
Assume that ${\bf p} \in \mathcal{V}_\mu$ with $\mu \in \mathbb{R}_+$, then
\begin{equation}\label{eq:Gini_Wasserstein_summary}
W_1({\bf p},{\bf p}^*) \leq \begin{cases}
2\,\mu\,\left(G[{\bf p}] - G[{\bf p}^*]\right), & \textrm{if ~$\mu \in \mathbb{N}_+$} \\
\frac{2\,\mu}{\min\{\mu-\floor*{\mu},\floor*{\mu}+1-\mu\}}\,\left(G[{\bf p}] - G[{\bf p}^*]\right), & \textrm{if ~$\mu \in \mathbb{R}_+ \setminus \mathbb{N}_+$}
\end{cases}
\end{equation}
where $W_1({\bf p},{\bf p}^*)$ represents the Wasserstein distance (of order $1$) \cite{santambrogio_optimal_2015} between ${\bf p}$ and the shifted Bernoulli distribution ${\bf p}^*$.
\end{theorem}

On the other hand, sometimes one also encounters the situation where the Gini index is monotone increasing along the solution of the system \eqref{eq:mean-field} (observed numerically for the system described in Example \ref{ex:1} starting from certain initial datum), i.e., we have
\begin{equation}\label{eq:Gini_decay}
\frac{\dd}{\dd t} G[{\bf p}] \geq 0,
\end{equation}
leading to accentuated wealth inequality as time evolves (the so-called ``rich get richer'' phenomenon). In the extreme case where $G[{\bf p}] \to 1$ so that the wealth inequality approaches to its maximum possible value, we expect that the distribution ${\bf p}$ converges (in some sense) to a Dirac delta $\delta_0 \coloneqq (1,0,0,\ldots,0,\ldots)$ centered at $0$, which implies the appearance of the so-called ``oligarchy'' \cite{boghosian_oligarchy_2017} where a vanishing portion of agents own the entire fortune of the economic society (which is again observed numerically in the ODE model given in Example \ref{ex:1}). We aim to quantify the closeness of ${\bf p}$ towards the singular Dirac mass centered at $0$, in the setting that $G[{\bf p}] \to 1$ (numerically observed and analytically justified in many models from econophysics \cite{boghosian_h_2015,boghosian_oligarchy_2017,cohen_bounding_2023} where the state space is $\mathbb{R}_+$ instead of $\mathbb N$). In a nutshell, we provide a answer to the following question: Assume that one can prove the convergence of the Gini index to its maximum possible value in the sense that
\begin{equation}\label{eq:conv_Gini}
G[{\bf p}(t)] \xrightarrow{t \to \infty} 1.
\end{equation}
In which sense can we establish the convergence of ${\bf p}(t)$ towards the Dirac delta $\delta_0$ centered at $0$~? Our result is stated as follows:
\begin{theorem}\label{thm:2}
Assume that ${\bf p} \in \mathcal{V}_\mu$ with $\mu \in \mathbb{R}_+$, then
\begin{equation}\label{eq:Gini_Wasserstein_summary}
\|{\bf p} - \delta_0\|_{\ell^1} \leq 2\,\sqrt{\mu}\,\sqrt{1-G[{\bf p}]}.
\end{equation}
In particular, $\|{\bf p} - \delta_0\|_{\ell^1} \to 0$ if $G[{\bf p}] \to 1$.
\end{theorem}

\section{Proof of main results}

\subsection{Proof of Theorem \ref{thm:main}}

This section is devoted to the proof of Theorem \ref{thm:main}. We split the proof into two parts (for the ease of presentation), depending on the whether $\mu$ is an integer or not. We first establish the following result when $\mu \in \mathbb{N}$ is an integer.
\begin{proposition}\label{prop:1}
Assume that ${\bf p} \in \mathcal{V}_\mu$ with $\mu \in \mathbb{N}_+$, then
\begin{equation}\label{eq:Gini_Wasserstein}
W_1({\bf p},\delta_\mu) = W_1({\bf p},{\bf p}^*) \leq 2\,\mu\,G[{\bf p}] = 2\,\mu\,\left(G[{\bf p}] - G[{\bf p}^*]\right),
\end{equation}
where $\delta_\mu$ denotes the Dirac delta distribution centered at $\mu$.
\end{proposition}

\begin{proof}
We first establish a weaker version of the bound \eqref{eq:Gini_Wasserstein} using a rather elementary (probabilistic) approach, which reads as
\begin{equation}\label{eq:weak_bound}
W_1({\bf p},\delta_\mu) \leq 2\,\sqrt{2}\,\mu\,\sqrt{G[{\bf p}]}.
\end{equation}
Indeed, let $X \sim {\bf p}$ be a $\mathbb{N}$-valued random variable distributed according to the law ${\bf p}$. Then
\begin{equation}\label{eq:weak_1}
\begin{aligned}
W_1({\bf p},\delta_\mu) &= \mathbb{E}|X-\mu| = \mathbb{E}\left[|\sqrt{X}-\sqrt{\mu}|\cdot|\sqrt{X}+\sqrt{\mu}|\right] \\
&\leq \sqrt{\mathbb{E}|\sqrt{X}-\sqrt{\mu}|^2}\,\sqrt{\mathbb{E}|\sqrt{X}+\sqrt{\mu}|^2} \\
&\leq 2\,\sqrt{2}\,\sqrt{\mu}\,\sqrt{\Var[\sqrt{X}]},
\end{aligned}
\end{equation}
where the last inequality follows from Jensen's inequality $\mathbb{E}\sqrt{X} \leq \sqrt{\mathbb{E}X} = \sqrt{\mu}$ together with the observation that $\Var[\sqrt{X}] = \mu - \left(\mathbb{E}\sqrt{X}\right)^2$. On the other hand, we also have
\begin{equation}\label{eq:weak_2}
\begin{aligned}
2\,\mu\,G[{\bf p}] = \sum_{i\geq 0}\sum_{j \geq 0} |i-j|\,p_i\,p_j &\geq \sum_{i\geq 0}\sum_{j \geq 0} |\sqrt{i}-\sqrt{j}|^2\,p_i\,p_j \\
&=\sum_{i\geq 0}\sum_{j \geq 0} (i+j-2\,\sqrt{i}\,\sqrt{j})\,p_i\,p_j \\
&= 2\,\mu - 2\,\left(\mathbb{E}\sqrt{X}\right)^2 = 2\,\Var[\sqrt{X}].
\end{aligned}
\end{equation}
Assembling \eqref{eq:weak_1} and \eqref{eq:weak_2} together leads us to the advertised bound \eqref{eq:weak_bound}. In order to achieve the improved estimate provided by \eqref{eq:Gini_Wasserstein}, we make use of a generic result on the Gini index established in the recent work \cite{cao_iterative_2024}, that is, for all ${\bf p} \in \mathcal{V}_\mu$ with $\mu \in \mathbb{R}_+$ we have
\begin{equation}\label{eq:key}
\mu\,G[{\bf p}] \geq \max\left\{\sum_{j \geq \floor*{\mu}+1} (j-\mu)\,p_j, \sum_{i\leq \floor*{\mu}} (\mu-i)\,p_i\right\}.
\end{equation}
Let $F_n = \sum_{m\leq n} p_m$ for all $n\in \mathbb N$ and $F_{-1} \coloneqq 0$ to be the cumulative distribution function associated to the probability mass function ${\bf p}$, we deduce that
\begin{equation*}
\begin{aligned}
W_1({\bf p},\delta_\mu) &= \sum_{n\leq \mu-1} F_n + \sum_{n\geq \mu} (1-F_n) = \sum_{m\leq \mu-1} \sum_{n=m}^{\mu-1} p_m + \sum_{m > \mu} \sum_{n=\mu}^{m-1} p_m \\
&= \sum_{m\leq \mu-1} (\mu-m)\,p_m + \sum_{m>\mu} (m-\mu)\,p_m \leq 2\,\mu\,G[{\bf p}],
\end{aligned}
\end{equation*}
which completes the proof.
\end{proof}

\begin{remark}
Under the settings of Proposition \ref{prop:1} we can also obtain an upper bound on the Gini index $G[{\bf p}]$ in terms of the Wasserstein distance $W_1$ between ${\bf p} \in \mathcal{V}_\mu$ and the Dirac mass $\delta_\mu$. Indeed, we recall that an alternative definition of the Gini index $G[{\bf p}]$ is provided by
\begin{equation}\label{def2:Gini}
G[{\bf p}] = \frac{1}{2\,\mu}\,\mathbb{E}|X-X'|,
\end{equation}
where $X$ and $X'$ are i.i.d. random variables distributed according to ${\bf p} \in \mathcal{V}_\mu$. Thanks to the triangle inequality we deduce that \begin{equation}\label{eq:Wasserstein_Gini}
2\,\mu\,G[{\bf p}] = \mathbb{E}|X-X'| \leq 2\,\mathbb{E}|X-\mu| = 2\,W_1({\bf p},\delta_\mu).
\end{equation}
\end{remark}

Now, as long as $\mu \in \mathbb{R}_+ \setminus \mathbb{N}_+$ is not an integer, the situation becomes more complicated since ${\bf p}^* \neq \delta_\mu$ and
\begin{equation}\label{eq:Gini_equilibrium}
G[{\bf p}^*] = \frac{1}{\mu}\,p^*_{\floor*{\mu}}\,p^*_{\floor*{\mu}+1} = \frac{1}{\mu}\,(1-\mu + \floor*{\mu})\,(\mu - \floor*{\mu}) > 0.
\end{equation}
We first treat a special and simpler case where $\mu \in (0,1)$, so that $\floor*{\mu} = 0$ and $G[{\bf p}^*] = 1-\mu$. In this case, we have
\begin{equation*}
\begin{aligned}
W_1({\bf p},{\bf p}^*) &= |p_0-(1-\mu)| + \sum_{n\geq 1} (1-F_n) = \mu-(1-p_0) + \sum_{n\geq 1} \sum_{k\geq n+1} p_k \\
&= \mu-(1-p_0) + \sum_{k\geq 2} \sum_{1\leq n\leq k-1} p_k = \mu-(1-p_0) + \sum_{k\geq 2} (k-1)\,p_k \\
&= \mu-(1-p_0) + \sum_{k\geq 1} (k-1)\,p_k = 2\,[\mu - (1-p_0)].
\end{aligned}
\end{equation*}
On the other hand, we also have
\begin{equation*}
G[{\bf p}] = \frac{1}{2\,\mu} \sum\limits_{i\geq 0}\sum\limits_{j\geq 0} |i-j|\,p_i\,p_j \geq \frac{1}{\mu}\,p_0\,\sum\limits_{j\geq 0} j\,p_j = p_0,
\end{equation*}
leading us to
\begin{equation*}
G[{\bf p}] - G[{\bf p}^*] \geq p_0 - (1-\mu) = \mu-(1-p_0) = \frac 12\,W_1({\bf p},{\bf p}^*).
\end{equation*}
In summary, we have proved the following result:
\begin{lemma}\label{lem:1}
Assume that ${\bf p} \in \mathcal{V}_\mu$ with $\mu \in (0,1)$, then
\begin{equation}\label{eq:Gini_Wasserstein_mu01}
W_1({\bf p},{\bf p}^*) \leq 2\,\left(G[{\bf p}] - G[{\bf p}^*]\right).
\end{equation}
\end{lemma}

We now generalize the strategy behind the proof of Lemma \eqref{lem:1} to a generic $\mu \in \mathbb{R}_+ \setminus \mathbb{N}_+$. The proof of the following proposition, together with the proof of Proposition \ref{prop:1} given above, allows us to conclude the proof of Theorem \ref{thm:main}.
\begin{proposition}\label{prop:2}
Assume that ${\bf p} \in \mathcal{V}_\mu$ with $\mu \in \mathbb{R}_+ \setminus \mathbb{N}_+$, then
\begin{equation}\label{eq:Gini_Wasserstein_general}
W_1({\bf p},{\bf p}^*) \leq \frac{2\,\mu}{\min\{\mu-\floor*{\mu},\floor*{\mu}+1-\mu\}}\,\left(G[{\bf p}] - G[{\bf p}^*]\right).
\end{equation}
\end{proposition}

\begin{proof}
The proof consists of establishing a suitable upper bound on $W_1({\bf p},{\bf p}^*)$ along with a comparable lower bound for $G[{\bf p}] - G[{\bf p}^*]$. We first prove the following upper bound on $W_1({\bf p},{\bf p}^*)$:
\begin{equation}\label{eq:W1_upper_bound}
W_1({\bf p},{\bf p}^*) \leq 2\,\max\left\{\sum\limits_{n\leq \floor*{\mu}} (\floor*{\mu}+1-n)\,p_n - (\floor*{\mu}+1-\mu), \sum\limits_{n\leq \floor*{\mu}} (\floor*{\mu}-n)\,p_n\right\}.
\end{equation}
Indeed, we can write
\begin{equation*}
\begin{aligned}
W_1({\bf p},{\bf p}^*) &= \sum_{n\leq \floor*{\mu}-1} F_n + \left|F_{\floor*{\mu}} - (\floor*{\mu}+1-\mu)\right| + \sum_{n\geq \floor*{\mu}+1} (1-F_n) \\
&= \sum_{n\leq \floor*{\mu}-1} (\floor*{\mu}-n)\,p_n + \left|F_{\floor*{\mu}} - (\floor*{\mu}+1-\mu)\right| + \sum_{n\geq \floor*{\mu}+1} \sum_{m\geq n+1}p_m\\
&= \sum_{n\leq \floor*{\mu}-1} (\floor*{\mu}-n)\,p_n + \left|F_{\floor*{\mu}} - (\floor*{\mu}+1-\mu)\right| + \sum_{m \geq \floor*{\mu}+1} (m-\floor*{\mu}-1)\,p_m
\end{aligned}
\end{equation*}
and notice that the last summand can be recast as follows:
\begin{equation*}
\begin{aligned}
\sum_{m \geq \floor*{\mu}+1} (m-\floor*{\mu}-1)\,p_m &= \mu - \sum_{n\leq \floor*{\mu}} n\,p_n - (\floor*{\mu}+1)\,\big(1-\sum_{n\leq \floor*{\mu}} p_n\big) \\
&= \sum_{n\leq \floor*{\mu}} (\floor*{\mu}+1-n)\,p_n - (\floor*{\mu}+1-\mu).
\end{aligned}
\end{equation*}
Therefore, if $p_{\floor*{\mu}} \geq (\floor*{\mu}+1-\mu) - \sum_{n\leq \floor*{\mu}-1} p_n$ so that $F_{\floor*{\mu}} \geq (\floor*{\mu}+1-\mu)$, then
\begin{equation}\label{eq:W1_e1}
W_1({\bf p},{\bf p}^*) = 2\,\left[\sum_{n\leq \floor*{\mu}} (\floor*{\mu}+1-n)\,p_n - (\floor*{\mu}+1-\mu)\right].
\end{equation}
Otherwise if $p_{\floor*{\mu}} < (\floor*{\mu}+1-\mu) - \sum_{n\leq \floor*{\mu}-1} p_n$ so that $F_{\floor*{\mu}} < (\floor*{\mu}+1-\mu)$, then
\begin{equation}\label{eq:W1_e2}
W_1({\bf p},{\bf p}^*) = 2\,\sum_{n\leq \floor*{\mu}} (\floor*{\mu}-n)\,p_n.
\end{equation}
Combining \eqref{eq:W1_e1} and \eqref{eq:W1_e2} then gives rise to the estimate \eqref{eq:W1_upper_bound}. We now claim that
\begin{equation}\label{eq:Gini_lower_bound}
\mu\,\left(G[{\bf p}] - G[{\bf p}^*]\right) \geq C_\mu\,\left(\sum\limits_{n\leq \floor*{\mu}} (\floor*{\mu}+1-n)\,p_n - (\floor*{\mu}+1-\mu)+ \sum\limits_{n\leq \floor*{\mu}} (\floor*{\mu}-n)\,p_n\right)
\end{equation}
with $C_\mu \coloneqq \min\{\mu-\floor*{\mu},\floor*{\mu}+1-\mu\}$, which is sufficient to derive the advertised bound \eqref{eq:Gini_Wasserstein_general} by taking into account of the estimate \eqref{eq:W1_upper_bound}. To establish \eqref{eq:Gini_lower_bound} we bound $\mu\,G[{\bf p}]$ as follows:
\begin{equation*}
\begin{aligned}
\mu\,G[{\bf p}] &\geq \sum_{n\leq \floor*{\mu}} p_n\,\sum_{j\geq 0} |j-n|\,p_j \\
&= \sum_{n\leq \floor*{\mu}-1} p_n\,\sum_{j\geq 0} |j-n|\,p_j + p_{\floor*{\mu}}\,\left(\mu-\floor*{\mu}+ 2\,\sum_{\ell \leq \floor*{\mu}-1} (\floor*{\mu}-\ell)\,p_\ell\right) \\
&= \sum_{n\leq \floor*{\mu}-1} p_n\,\left(\mu-n + 2\,\sum_{\ell \leq n-1} (n-\ell)\,p_\ell\right) + p_{\floor*{\mu}}\,\left(\mu-\floor*{\mu}+ 2\,\sum_{\ell \leq \floor*{\mu}-1} (\floor*{\mu}-\ell)\,p_\ell\right) \\
&\geq \sum_{n\leq \floor*{\mu}-1} p_n\,(\mu-n) + p_{\floor*{\mu}}\,\left(\mu-\floor*{\mu}+ 2\,\sum_{\ell \leq \floor*{\mu}-1} (\floor*{\mu}-\ell)\,p_\ell\right).
\end{aligned}
\end{equation*}
Due to the expression of $G[{\bf p}^*]$ \eqref{eq:Gini_equilibrium}, we obtain
\begin{equation*}
\begin{aligned}
\mu\,\left(G[{\bf p}] - G[{\bf p}^*]\right) &\geq (\mu-\floor*{\mu})\,\left(p_{\floor*{\mu}}-(\floor*{\mu}+1-\mu)\right) + \sum_{n\leq \floor*{\mu}-1} p_n\,(\mu-n) \\
&=(\mu-\floor*{\mu})\,\left(\sum_{n\leq \floor*{\mu}} (\floor*{\mu}+1-n)\,p_n - (\floor*{\mu}+1-\mu)\right)\\
&\qquad + \sum_{n\leq \floor*{\mu}-1} p_n\,(\mu-n) - \sum_{n\leq \floor*{\mu}-1} (\mu-\floor*{\mu})\,(\floor*{\mu}+1-n)\,p_n \\
&= (\mu-\floor*{\mu})\,\left(\sum_{n\leq \floor*{\mu}} (\floor*{\mu}+1-n)\,p_n - (\floor*{\mu}+1-\mu)\right) \\
&\qquad + (\floor*{\mu}+1-\mu)\,\sum_{n\leq \floor*{\mu}-1} (\floor*{\mu}-n)\,p_n \\
&\geq C_\mu\,\left(\sum\limits_{n\leq \floor*{\mu}} (\floor*{\mu}+1-n)\,p_n - (\floor*{\mu}+1-\mu)+ \sum\limits_{n\leq \floor*{\mu}} (\floor*{\mu}-n)\,p_n\right),
\end{aligned}
\end{equation*}
whence the proof is completed.
\end{proof}

\begin{remark}
We emphasize that the proof of Proposition \ref{prop:2} implies $G[{\bf p}] \geq G[{\bf p}^*]$, which yields an alternative (and direct) proof of the variational characterization ${\bf p}^* = \argmin_{{\bf p} \in \mathcal{V}_\mu} G[{\bf p}]$ of the distribution ${\bf p}^*$ proved in a recent work \cite{cao_iterative_2024} (through proof by contradiction).
\end{remark}

\subsection{Proof of Theorem \ref{thm:2}}

We record the proof of Theorem \ref{thm:2} in this section. The key ingredient lies in the following alternative expression for the Gini index \eqref{def1:Gini}.

\begin{lemma}\label{lem:Gini_alternative_def}
Assume that ${\bf p} \in \mathcal{V}_\mu$. Then
\begin{equation}\label{eq:Gini_expression2}
G[{\bf p}] = 1 - \frac{1}{\mu}\,\sum\limits_{n\geq 0} (1-F_n)^2
\end{equation}
\end{lemma}

\begin{proof}
We start from the definition of the Gini index \eqref{def1:Gini} and proceed as follows:
\begin{align*}
G[{\bf p}] &= \frac{1}{2\,\mu}\sum_{m\geq 0}\sum_{n\geq 0} |m-n|\,p_m\,p_n = \frac{1}{\mu}\sum_{m\geq 0}\sum_{n\geq m+1} (n-m)\,p_m\,p_n \\
&= \frac{1}{\mu}\sum_{m\geq 0}\sum_{n\geq m+1} n\,p_m\,p_n - \frac{1}{\mu}\sum_{m\geq 0}\sum_{n\geq m+1} m\,p_m\,p_n \\
&= \frac{1}{\mu}\sum_{n\geq 0} n\,p_n\,F_{n-1} - \frac{1}{\mu}\sum_{n\geq 0} n\,p_n\,(1-F_n) \\
&= \frac{1}{\mu}\left[\mu - 2\,\sum_{n\geq 0} n\,p_n\,(1-F_n) - \sum_{n\geq 0} n\,p^2_n\right] \\
&= 1 - \frac{1}{\mu}\,\sum\limits_{n\geq 0} (1-F_n)^2,
\end{align*}
where the last identity follows from the elementary observation that $2\,p_n\,(1-F_n) + p^2_n = (1-F_{n-1})^2 - (1-F_n)^2$ for all $n \in \mathbb N$. The aforementioned computations allows us to conclude the proof.
\end{proof}

\begin{proof}[Proof of Theorem \ref{thm:2}]
Thanks to Lemma \ref{lem:Gini_alternative_def}, we have
\begin{equation*}
1 - G[{\bf p}] = \frac{1}{\mu}\,\sum\limits_{n\geq 0} (1-F_n)^2 \geq \frac{1}{\mu}\,(1-F_0)^2 = \frac{1}{\mu}\,(1-p_0)^2,
\end{equation*}
from which we deduce that $1-p_0 \leq \sqrt{\mu\,(1 - G[{\bf p}])}$. Consequently,
\begin{equation*}
\|{\bf p} - \delta_0\|_{\ell^1} = (1-p_0) + \sum_{n\geq 1} p_n = 2\,(1-p_0) \leq 2\,\sqrt{\mu}\,\sqrt{1-G[{\bf p}]},
\end{equation*}
which finishes the proof of Theorem \ref{thm:2}.
\end{proof}

\begin{remark}
We remark here that under the settings of Theorem \ref{thm:2}, the Wasserstein distance $W_1({\bf p},\delta_0)$ between ${\bf p}$ and $\delta_0$ does not converge to zero when $G[{\bf p}] \to 1$. Indeed, we have
\begin{equation}\label{eq:W1_non_conv}
W_1({\bf p},\delta_0) = \sum_{n\geq 0} (1-F_n) = \sum_{n\geq 0}\sum_{m\geq n+1} p_m = \sum_{m\geq 1} m\,p_m = \mu
\end{equation}
whenever ${\bf p} \in \mathcal{V}_\mu$.
\end{remark}

\section{Conclusion}

In this manuscript, we investigated the implication of the convergence of Gini index (to its equilibrium value) shown numerically or analytically for solutions of certain infinite dimensional ODE systems (obtained from the mean-field analysis of certain stochastic agent-based models encountered in the econophysics and sociophysics literature) on the convergence of the underlying probability distributions. To the best of our knowledge, despite of the popularity and utility of the Gini index in the analysis of ODE systems of the form \eqref{eq:mean-field} and/or some Boltzmann-type kinetic PDEs (where the state space is continuous) \cite{boghosian_h_2015,boghosian_oligarchy_2017,cao_explicit_2021,cohen_bounding_2023}, it is quite surprising that little effort has been dedicated to study the consequence of such convergence at the level of Gini index. It is rather natural, from the application or modeling point of view, to speculate that convergence of Gini index (which often serves as a Lyapunov functional for the underlying evolution equation or system under consideration) should imply (at least) convergence in the sense of distributions. We filled in the aforementioned gap in this work and shed light on the relation between Gini index and other (perhaps much more) popular metrics for quantifying the closeness between probability distributions (such as Wasserstein distances).

To conclude, let us mention a more sophisticated task pertaining to the theme of the present manuscript. For many ODE systems (and Boltzmann-kinetic equations) mentioned throughout this paper, we would like to know the underlying reason that the Gini index will serve as a Lyapunov functional for those dynamics. The problem of construction of appropriate Lyapunov functionals for differential equations or system of differential equations has been subjected to substantial research across different branches of the applied mathematics \cite{hsu_survey_2005,hafstein_algorithm_2009}, and we speculate that the gradient flow approach developed recently for certain mean-field ODE systems on discrete spaces \cite{erbar_gradient_2016} might be helpful to uncover the important role played by the Gini index in the analysis of ODE systems with specific structures.


\begin{thebibliography}{99}

\bibitem{auricchio_extending_2024}
Gennaro Auricchio, Paolo Giudici, and Giuseppe Toscani.
\newblock Extending the Gini Index to Higher Dimensions via Whitening Processes.
\newblock {\em arXiv preprint arXiv:2409.1011}, 2024.

\bibitem{bassetti_mean_2015}
Federico Bassetti, and Giuseppe Toscani.
\newblock Mean field dynamics of interaction processes with duplication, loss and copy.
\newblock \emph{Mathematical Models and Methods in Applied Sciences}, 25(10):1887--1925, 2015.

\bibitem{boghosian_h_2015}
Bruce M.Boghosian, Merek Johnson, and Jeremy A. Marcq.
\newblock An $H$ Theorem for Boltzmann’s Equation for the Yard-Sale Model of Asset Exchange: The Gini Coefficient as an $H$ Functional.
\newblock {\em Journal of Statistical Physics}, 161:1339--1350, 2015.

\bibitem{boghosian_oligarchy_2017}
Bruce M. Boghosian, Adrian Devitt-Lee, Merek Johnson, Jie Li, Jeremy A. Marcq, and Hongyan Wang.
\newblock Oligarchy as a phase transition: The effect of wealth-attained advantage in a Fokker--Planck description of asset exchange.
\newblock {\em Physica A: Statistical Mechanics and its Applications}, 476:15--37, 2017.

\bibitem{cao_binomial_2022}
Fei Cao, and Nicholas F. Marshall.
\newblock From the binomial reshuffling model to Poisson distribution of money.
\newblock {\em Networks and Heterogeneous Media}, 19(1):24--43, 2024.

\bibitem{cao_derivation_2021}
Fei Cao, and Sebastien Motsch.
\newblock Derivation of wealth distributions from biased exchange of money.
\newblock {\em Kinetic \& Related Models}, 16(5):764--794, 2023.

\bibitem{cao_explicit_2021}
Fei Cao.
\newblock Explicit decay rate for the Gini index in the repeated averaging model.
\newblock {\em Mathematical Methods in the Applied Sciences}, 46(4):3583--3596, 2023.

\bibitem{cao_interacting_2022}
Fei Cao, and Pierre-Emannuel Jabin.
\newblock From interacting agents to Boltzmann-Gibbs distribution of money.
\newblock {\em arXiv preprint arXiv:2208.05629}, 2022.

\bibitem{cao_iterative_2024}
Fei Cao, and Stephanie Reed.
\newblock The iterative persuasion-polarization opinion dynamics and its mean-field analysis.
\newblock {\em arXiv preprint arXiv:2408.00148}, 2024.

\bibitem{cao_quantitative_2024}
Fei Cao, and Jincheng Yang.
\newblock Quantitative convergence guarantees for the mean-field dispersion process.
\newblock {\em arXiv preprint arXiv:2406.05043}, 2024.

\bibitem{cao_sticky_2024}
Fei Cao, and Sebastien Motsch.
\newblock Sticky dispersion on the complete graph: a kinetic approach.
\newblock {\em arXiv preprint arXiv:2404.08868}, 2024.

\bibitem{cao_uncovering_2022}
Fei Cao, and Sebastien Motsch.
\newblock Uncovering a two-phase dynamics from a dollar exchange model with bank and debt.
\newblock {\em SIAM Journal on Applied Mathematics}, 83(5):1872--1891, 2023.

\bibitem{cao_uniform_2024}
Fei Cao, and Roberto Cortez.
\newblock Uniform propagation of chaos for a dollar exchange econophysics model.
\newblock {\em European Journal of Applied Mathematics}, 1--13, 2024.

\bibitem{cohen_bounding_2023}
David W. Cohen, and Bruce M. Boghosian.
\newblock Bounding the approach to oligarchy in a variant of the yard-sale model.
\newblock {\em arXiv preprint arXiv:2310.16098}, 2023.

\bibitem{cooper_dispersion_2018}
Colin Coopery, Andrew McDowellz, Tomasz Radzikx, and Nicol{\'a}s Rivera.
\newblock Dispersion processes.
\newblock {\em Random Structures \& Algorithms}, 53(4):561--585, 2018.

\bibitem{daniya_classification_2020}
T. Daniya, Mithra Geetha, and K. Suresh Kumar.
\newblock Classification and regression trees with gini index.
\newblock {\em Advances in Mathematics: Scientific Journal}, 9(10):8237--8247, 2020.

\bibitem{de_dispersion_2023}
Umberto De Ambroggio, Tam{\'a}s Makai, and Konstantinos Panagiotou.
\newblock Dispersion on the complete graph.
\newblock {\em arXiv preprint arXiv:2306.02474}, 2023.

\bibitem{dragulescu_statistical_2000}
Adrian Dragulescu, and Victor~M. Yakovenko.
\newblock Statistical mechanics of money.
\newblock {\em The European Physical Journal B-Condensed Matter and Complex Systems}, 17(4):723--729, 2000.

\bibitem{erbar_gradient_2016}
Matthias Erbar, Max Fathi, Vaios Laschos, and Andr{\'e} Schlichting.
\newblock Gradient flow structure for McKean-Vlasov equations on discrete spaces.
\newblock {\em Discrete and Continuous Dynamical Systems}, 36(12):6799--6833, 2016.

\bibitem{hafstein_algorithm_2009}
Sigurdur Hafstein.
\newblock An algorithm for constructing Lyapunov functions.
\newblock {\em Electronic Journal of Differential Equations}, 08--101, 2009.

\bibitem{hsu_survey_2005}
Sze-Bi Hsu.
\newblock A survey of constructing Lyapunov functions for mathematical models in population biology.
\newblock {\em Taiwanese Journal of Mathematics}, 151--173, 2005.

\bibitem{hurley_comparing_2009}
Niall Hurley, and Scott Rickard.
\newblock Comparing measures of sparsity.
\newblock {\em IEEE Transactions on Information Theory}, 55(10):4723--4741, 2009.

\bibitem{meng_fair_2023}
Fanyong Meng, Dengyu Zhao, and Xumin Zhang.
\newblock A fair consensus adjustment mechanism for large-scale group decision making in term of Gini coefficient.
\newblock {\em Engineering Applications of Artificial Intelligence}, 126:106962, 2023.

\bibitem{santambrogio_optimal_2015}
Filippo Santambrogio.
\newblock Optimal transport for applied mathematicians.
\newblock \emph{Birk{\"a}user, NY}, 55(58-63):94, 2015.

\end{thebibliography}
\end{document}